\documentclass[a4paper,english]{article}

\usepackage{fullpage}
\usepackage{graphicx}
\usepackage{authblk}
\usepackage[utf8]{inputenc} 

\usepackage{xspace}
\usepackage{amsmath,amssymb,amsthm,amstext}
\usepackage{complexity} 
\usepackage{nicefrac} 
\usepackage{accents}  

\usepackage{subcaption}

\usepackage[shortlabels]{enumitem}
\setlist[enumerate]{itemsep=1pt,topsep=3pt}
\setlist[itemize]{itemsep=1pt,topsep=3pt}
\setlist[description]{itemsep=1pt,topsep=3pt}

\newcommand{\coinsonaline}{\textsc{CoinsOnAShelf}\xspace}
\newcommand{\threepartition}{\textsc{3-Partition}\xspace}

\newcommand{\footpoint}[1]{\underaccent{\dot}{#1}}
\newcommand{\fA}{\footpoint{A}}
\newcommand{\fB}{\footpoint{B}}
\newcommand{\fC}{\footpoint{C}}
\newcommand{\fD}{\footpoint{D}}
\newcommand{\fG}{\footpoint{G}}

\newcommand{\fY}{\footpoint{Y}}
\newcommand{\fZ}{\footpoint{Z}}
\newcommand{\order}{\mathcal{D}}

\let\geq\geqslant
\let\leq\leqslant

\usepackage{microtype} 

\newtheorem{theorem}{Theorem}
\newtheorem{lemma}[theorem]{Lemma}

\usepackage{hyperref}

\hypersetup{%
  breaklinks,%
  colorlinks=true,%
  linkcolor=[rgb]{0.45,0.0,0.0},%
  citecolor=[rgb]{0,0,0.45}
}

\makeatletter \ifx\showkeyslabelformat\@undefined\else
\renewcommand{\showkeyslabelformat}[1]{\normalfont\tiny\ttfamily#1}
\fi
\partopsep\z@ \textfloatsep 10pt plus 1pt minus 4pt
\def\section{\@startsection {section}{1}{\z@}{-3.5ex plus -1ex minus
    -.2ex}{2.3ex plus .2ex}{\large\bf}}
\def\subsection{\@startsection{subsection}{2}{\z@}{-3.25ex plus -1ex
    minus -.2ex}{1.5ex plus .2ex}{\normalsize\bf}}
\def\@fnsymbol#1{\ensuremath{\ifcase#1\or *\or 1\or 2\or 3\or 4\or
    5\or 6\or 7 \or 8\ or 9 \or 10\or 11 \else\@ctrerr\fi}}
\makeatother

\graphicspath{ {figs/} }

\title{Placing your Coins on a Shelf%
  \thanks{O.C.~is supported by NRF grant 2011-0030044 (SRC-GAIA) funded
    by the government of Korea.}} 

\author[1]{Helmut Alt}
\author[2]{Kevin Buchin}
\author[3]{Steven Chaplick}
\author[4]{Otfried Cheong}
\author[5]{Philipp Kindermann}
\author[6]{Christian Knauer}
\author[6]{Fabian Stehn}

\affil[1]{Freie Universität Berlin, Germany \texttt{alt@fu-berlin.de}}
\affil[2]{Technische Universiteit Eindhoven, Netherlands
  \texttt{k.a.buchin@tue.nl}}
\affil[3]{Universität Würzburg, Germany
  \texttt{steven.chaplick@uni-wuerzburg.de}}
\affil[4]{KAIST, Korea \texttt{otfried@kaist.airpost.net}}
\affil[5]{University of Waterloo, Canada
  \texttt{pkinderm@uwaterloo.ca}} 
\affil[6]{Universität Bayreuth, Germany
  \texttt{[christian.knauer|fabian.stehn]@uni-bayreuth.de}}

\begin{document}

\maketitle

\begin{abstract}
  We consider the problem of packing a family of disks ``on a shelf,''
  that is, such that each disk touches the $x$-axis from above and
  such that no two disks overlap. We study the problem of minimizing
  the distance between the leftmost point and the rightmost point of
  any disk in such a packing.  We show how to approximate this problem
  within a factor of $\nicefrac 43$ in~$O(n \log n)$ time. We further
  provide an $O(n \log n)$-time exact algorithm for a special case
  which includes inputs where the ratio between the largest radius and
  the smallest radius is less than four.  On the negative side, we
  prove that the problem is \NP-hard even when the ratio between the
  largest radius and the smallest radius is at most~36.
\end{abstract}

\section{Introduction}

Packing problems have a long history and abundant literature. Circular
disks and spherical balls, because of their symmetry and simplicity,
are of particular interest from a theoretical point of view.
Historically, Johannes Kepler conjectured that an optimal packing of
unit spheres into the Euclidean three-space cannot have greater
density than the face-centered cubic packing~\cite{kepler1611}.  The
conjecture was first proven to be correct by Hales and
Ferguson~\cite{hales06}. A more recent treatment of the proof is given
by Hales et al.~\cite{hales17}.  The proof of the 2-dimensional
version of Kepler's conjecture, that is, packing unit disks into the
Euclidean two-space, is elementary and attributed to Lagrange (1773).

Packing unit disks into 2-dimensional shapes in the plane is a well
studied problem in recreational mathematics. Croft et
al.~\cite{croft91} give an overview of packing geometrical objects in
finite-sized containers, for instance finding the smallest square
(circle, isosceles triangle, etc.) such that a given number of~$n$
unit disks can be packed into it.  Specht~\cite{specht} presents the
best known packings of up to $10,000$ disks into various containers.

Algorithmically, many packing problems are \NP-hard, some are not even
known to be in~\NP. Demaine, Fekete, and Lang showed that the problems
whether a given set of circular disks of arbitrary radii can be packed
into a given square, rectangle, or triangle are all \NP-hard
problems~\cite{abs-1008-1224}.

\begin{figure}[t]
  \centerline{\includegraphics[scale=0.8]{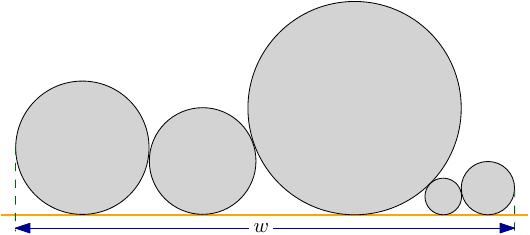}}
  \caption{Illustration of the span $w$ of a valid (but not optimal)
    placement of five discs.}
  \label{figure:illustrationWidth}
\end{figure}

We will discuss a particular ``nearly'' one-dimensional packing
problem for disks from an algorithmic perspective.  We are given a
family of disks that we wish to arrange ``on a shelf,'' that is, such
that each disk touches the $x$-axis from above and such that no two
disks overlap; see Figure~\ref{figure:illustrationWidth}.  The goal is
to minimize the \emph{span} of the resulting configuration, that is,
to minimize the horizontal distance between the leftmost and the
rightmost point of any disk. In other words, we want to minimize the
required width of the shelf. Obviously, this problem is trivial for
unit disks, so we allow the disks to have different sizes.

\subparagraph{Related work.}

D\"urr et al.~\cite{Duerr2017} independently study
shelf packings, but for the case when the objects are isosceles 
right-angle triangles (instead of disks). Namely, given $n$
sizes of this triangle, they ask for the shortest horizontal span in
which the triangles can be arranged so that their lowest point lies on
the $x$-axis, while the triangles do not overlap.  Their entirely
independent results are quite similar to ours: an \NP-hardness proof
by reduction from \threepartition, a fast algorithm for a special
case, and a $\nicefrac 32$-approximation algorithm.

Klemz et al.~\cite{klemz} show that it is \NP-hard to decide if $n$
given disks fit around a large center disk, such that each disk is in
contact with the center disk while all disks are disjoint.  Their
proof is by reduction from \threepartition as well.

Stoyan and Yaskov~\cite{stoyan2004} introduce the problem of
packing disks of unequal sizes into a strip of given height
and minimizing the required width which is known as the
\emph{circular open dimension problem}.

Miyazawa et al.~\cite{10.1007/978-3-662-44777-2_59} consider the problem of packing a set of circles into a minimum number of unit square bins. They give an asymptotic approximation scheme (APTAS) when resource augmentation in one dimension is allowed (i.e., they use bins of height slightly larger than one). They also obtain an APTAS for the circle strip packing problem, where the objective is to pack a set of circles into a strip of unit width and minimum height.

Lintzmayer et al.~\cite{10.1007/978-3-319-77404-6_54} present a polynomial-time approximation scheme for the Two-dimensional Knapsack for Circles problem, where one is given a set of circles and the goal is to pack a subset of them into a rectangular bin of fixed dimensions such that the sum of the area of the packed circles is maximum.

\subparagraph{Our results.} 

We first give some useful definitions and properties for touching
disks in Section~\ref{section:preliminaries}. 
The hardness of the problem arises from 
the fact that disks can sometimes ``hide'' in the holes formed
by larger disks, as in Figure~\ref{figure:generalCase}.  For this
reason, in Section~\ref{section:linearCase}, we consider the
special case where, for any ordering of the disks, each disk can touch
only its left and its right neighbor (where the two walls bounding the
span count as neighbors as well).  In particular, this implies that no
disk will ever fit in a gap between two other disks. We call this the
\emph{linear case}, see Figure~\ref{figure:linearCase}.
It turns out that for this (linear) case the optimal configuration
depends \emph{only} on the relative order of the disk
sizes,\footnote{The median disk for an odd number of disks is the only
  exception, it can be on either end, depending on its actual size.}
so it suffices to sort the disks in $O(n \log n)$ time to determine
the optimal sequence.

In Section~\ref{section:generalCase}, we show that in its general form, the
problem is \NP-hard.  More precisely, we show that given $n$ disk sizes
and a number $\delta > 0$, it is \NP-hard to decide if a
non-overlapping arrangement of the disks with horizontal span at
most~$\delta$ exists.  Our \NP-hardness proof is by a reduction from
\threepartition, and exploits the fact that disks can ``hide'' in the
holes formed by larger disks.

Finally, in Section~\ref{section:approximation}, we give an approximation
algorithm that runs in $O(n \log n)$~time and guarantees a span at
most $\nicefrac 43$ times the optimal span.

\begin{figure}[htb]
  \begin{subfigure}[t]{.48\linewidth}
    \centerline{\includegraphics[page=1]{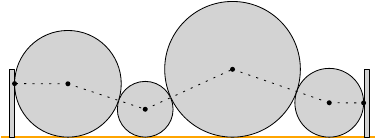}}
    \caption{The linear case.}
    \label{figure:linearCase}
  \end{subfigure}
  \hfill
  \begin{subfigure}[t]{.48\textwidth}
    \centerline{\includegraphics[page=2]{caseDistinction}}
    \caption{Small disks can ``hide'' between larger disks.}
    \label{figure:generalCase}
  \end{subfigure}
  \caption{Illustration of different instances of the problem.}
\end{figure}

\section{Preliminaries}\label{section:preliminaries}

For reasons that will become obvious shortly, it will be convenient to
define the \emph{size} of a disk as the \emph{square root} of its
radius.  We will denote disks by capital letters, and their size by
the corresponding lower-case letter. Namely, disk~$A$ has size~$a$,
radius~$a^{2}$, and diameter~$2a^{2}$. 

In a valid placement, each disk~$A$ touches the $x$-axis in its lowest
point.  We will call this point the \emph{footpoint} of the disk and
denote it~$\fA$.  All of our arguments are based on calculations
involving the distances between footpoints, so we start with the
following lemma.
\begin{lemma}
  \label{lemma:footpoint}
  If $A$ and $B$ touch, then their footpoint distance $\fA\fB$ is
  $2ab$.
\end{lemma}
\begin{proof}
  The statement holds for $a = b$, so we assume $a > b$ and consider
  the right-angled triangle with edge lengths $\fA\fB$, $a^{2} + b^{2}$, and
  $a^{2}-b^{2}$, see Figure~\ref{figure:introductionSize}.  We obtain
  $(\fA\fB)^{2} = (a^{2} + b^{2})^2 - (a^{2} - b^{2})^2 =
  4a^{2}b^{2}$.
\end{proof}
\begin{figure}[hbt]
  \centerline{\includegraphics{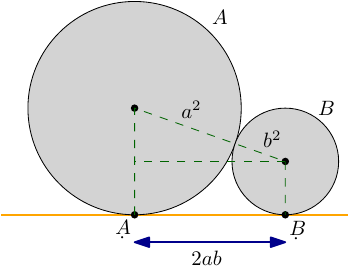}}
  \caption{The footpoint distance of two touching disks.}
  \label{figure:introductionSize}
\end{figure}

\begin{lemma}
  \label{lemma:gap}
  Let $G$ be the largest disk that fits in the gap formed by two
  touching disks~$A$ and~$B$. Then~$\nicefrac 1g = \nicefrac 1a + \nicefrac 1b$.
\end{lemma}
\begin{proof}
  Since $G$ is the largest disk that fits in the gap, it must touch
  both~$A$ and~$B$. By Lemma~\ref{lemma:footpoint} we have $2ab =
  \fA\fB = \fA\fG + \fG\fB = 2ag + 2gb$, proving the lemma.
\end{proof}

\begin{lemma}
  \label{lemma:wallgap}
  Let $G$ be the largest disk that fits in the gap between a disk~$A$
  and the vertical wall through~$A$'s rightmost point.
  Then $g = (\sqrt 2 -1) \cdot a$.
\end{lemma}
\begin{proof}
  Again, $G$ must touch both~$A$ and the wall, so we have $a^{2} =
  \fA\fG + g^{2} = 2ag + g^{2}$.  The positive solution to
  $g^{2} + 2ag - a^{2} = 0$ is $(\sqrt{2} - 1)\cdot a$.
\end{proof}

In any valid placement of the disks, their footpoints are distinct.
Thus, the footpoints induce a linear left-to-right order on the
disks.
We refer to this linear order as the \emph{footpoint sequence} of a
valid placement.
Further, disks are called \emph{consecutive} or \emph{neighbors}
when their footpoints are consecutive in the footpoint sequence.

\section{The Linear Case}
\label{section:linearCase}

In this section, we consider \emph{linear case instances}, that is,
instances where in any valid placement only consecutive pairs of disks
can touch, only the first disk (with the leftmost footpoint) touches
the left wall, and only the last disk touches the right wall.

By Lemmas~\ref{lemma:gap} and~\ref{lemma:wallgap}, this is true if and
only if the following condition holds: Let~$A$ be the largest disk,
$B$~the second largest, and~$Z$ the smallest disk in the
collection. Then $\nicefrac 1z < \nicefrac 1a + \nicefrac 1b$, and $z > (\sqrt{2} -1)\cdot a$.  The
condition holds in particular if the ratio between the largest and
smallest disk size is less than two (that is, if the ratio of
diameters is less than four), since then we have $\nicefrac 1z < \nicefrac 2a \leq \nicefrac 1a +
\nicefrac 1b$ and $z> \nicefrac a2 > (\sqrt 2 -1)\cdot a$.

In an optimal placement of a linear case instance, each disk must
touch both its neighbors.  Thus, the ordering of the disks uniquely
determines the exact placement of every disk in any layout of minimal
span.
From now on, we represent placements by the \emph{ordering} of the
disks, with the understanding that the placement minimizes the span
for this ordering. It remains to determine the optimal ordering.  We
will first give a lemma that allows us to improve a given ordering.
\begin{lemma}
  \label{lemma:reverse}
  Let $\order$ be a left-to-right or right-to-left ordering of the
  disks in a linear case instance.  Let $A$, $B$, $Z$ be three disks
  that appear in this order in~$\order$ such that $AB$ is a
  consecutive pair.  Let $\order'$ be the ordering obtained
  from~$\order$ by reversing the subsequence from~$B$ to~$Z$.  Then
  $\order'$ has smaller span than~$\order$ if one of the following is
  true:
  \begin{enumerate}[label=(C\arabic*),leftmargin=4em]
  \item\label{enum:reverse-1}
    $Z$ is the last disk and $a > b > z$;
  \item\label{enum:reverse-2}
    $Z$ is the last disk and $a < b < z$;
  \item\label{enum:reverse-3} $a > y$ and $b > z$, where $Y$ is the
    disk after~$Z$ in~$\order$;
  \item\label{enum:reverse-4} $a < y$ and $b < z$, where $Y$ is the
    disk after~$Z$ in~$\order$.
  \end{enumerate}
\end{lemma}
\begin{proof}
  First, suppose that $Z$ is the last disk in~$\order$. Then, except
  for $\fA\fB$ being replaced by $\fA\fZ$, each consecutive footpoint
  distance in $\order'$ is the same as in~$\order$. So, since the last
  disk in~$\order'$ is~$B$, the change in span is $\fA\fZ + b^{2} -
  \fA\fB - z^{2} = 2az + b^{2} - 2ab - z^{2} = (b + z - 2a)(b - z)$.
  For both $a < b < z$ and $a > b > z$, this is negative, and so
  $\order'$ has smaller span than~$\order$.

  Now suppose $Z$ is not the last disk, and let $Y$ be the disk
  after~$Z$. Here, except for $\fA\fB$ being replaced by $\fA\fZ$ and
  $\fZ\fY$ being replaced by $\fB\fY$, each consecutive footpoint
  distance in $\order'$ is the same as in~$\order$. Thus, the change
  in span is $\fA\fZ + \fB\fY - \fA\fB - \fZ\fY = 2(az + by - ab - zy)
  = 2(a-y)(z-b)$.  For $a > y$ and $b > z$ or $a < y$ and $b < z$,
  this is negative.  So, again $\order'$ has smaller span
  than~$\order$.
\end{proof}

We label a given family of $n$ disks in order of
decreasing size as $D_1, D_2, D_3, \dots, D_n$, and in order of
increasing size as $S_1, S_2, S_3, \dots, S_n$.  In other words, $d_1
\geq d_2 \geq d_3 \geq \dots \geq d_n$ and $s_1 \leq s_2 \leq s_3 \leq
\dots \leq s_n$.  Thus, each disk has two names, and we have $D_1 = S_n$,
$D_2 = S_{n-1}$, and so on until $D_n = S_1$.

We now prove our claim about the structure of the optimal
ordering (see also Figure~\ref{figure:lemma-linear}):
\begin{lemma}
  \label{lemma:ordering}
  Let $k$ be an integer with $1 \leq k \leq \nicefrac n2$.  In any
  optimal placement of $n$ disks with distinct sizes in a linear case
  instance, there is a consecutive subsequence of~$2k$ disks that
  consists of the~$k$ largest disks $D_1,\dots,D_k$ and the~$k$
  smallest disks $S_1,\dots,S_k$, and that is terminated by the
  disks~$S_k$ and~$D_k$.  If $k > 1$, then $D_kS_{k-1}$ and
  $S_{k}D_{k-1}$ are consecutive pairs.
\end{lemma}
\begin{proof}
  We use induction over~$k$. For~$k=1$, it suffices to prove that
  $S_1$ and $D_1$ are consecutive, so assume for a contradiction that
  this is not the case. Let $A = D_1$, $Z=S_1$, assume $A$ is to the
  left of~$Z$, and let $B$ be the right neighbor of~$A$.  By
  Lemma~\ref{lemma:reverse} (Case~\ref{enum:reverse-1}
  or~\ref{enum:reverse-3}), the sequence can now be improved by
  reversing the subsequence from~$B$ up to~$Z$.

  Assume now that $k > 1$ and that the statement holds for~$k-1$. This
  means that there is a consecutive subsequence of the disks~$\{S_1,
  \dots, S_{k-1}, D_1, \dots, D_{k-1}\}$, terminated by disk~$S_{k-1}$
  at the, say, right end and disk~$D_{k-1}$ at the left end, as in the
  example of Figure~\ref{figure:lemma-linear}.
  \begin{figure}[t]
    \centerline{\includegraphics[scale=0.8]{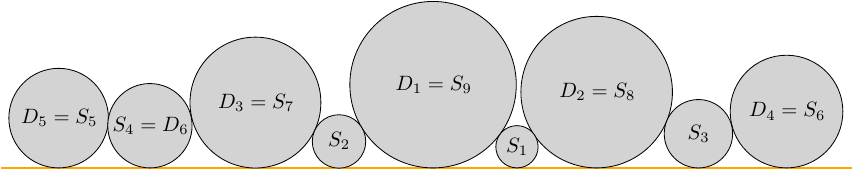}}
    \caption{An optimal placement in the linear case. For instance for
      $k=2$, the disks in $\{S_1,S_2,D_1,D_2\}$ form the consecutive
      subsequence starting with $S_2$ and ending with $D_2$.} 
    \label{figure:lemma-linear}
  \end{figure}

  We first show that the right neighbor of~$S_{k-1}$ is~$D_k$.  Assume
  this is not the case. We distinguish four cases:
  \begin{enumerate}[(1)]
  \item If $D_{k}$ appears to the right of $S_{k-1}$ (but not
    immediately adjacent), then we apply Lemma~\ref{lemma:reverse}
    (Case~\ref{enum:reverse-2} or~\ref{enum:reverse-4}) with $A =
    S_{k-1}$, $B$ the right neighbor of~$S_{k-1}$, and~$Z = D_{k}$.
  \item If $D_{k}$ appears to the left of $S_{k-1}$, then it must
    appear to the left of~$D_{k-1}$. If $D_{k}$ is not the left
    neighbor of~$D_{k-1}$, then apply Lemma~\ref{lemma:reverse}
    (Case~\ref{enum:reverse-1} or~\ref{enum:reverse-3}) with $A =
    D_{k}$, $B$ the right neighbor of~$D_{k}$, and $Z = S_{k-1}$.
  \item If $D_{k}$ is the left neighbor of~$D_{k-1}$ and $S_{k-1}$ is
    not the rightmost disk, then apply Lemma~\ref{lemma:reverse}
    (Case~\ref{enum:reverse-3}) with $A = D_{k}$, $B = D_{k-1}$, and
    $Z = S_{k-1}$.
  \item If $D_{k}$ is the left neighbor of~$D_{k-1}$ and $S_{k-1}$ is
    the rightmost disk, then $S_{k}$ appears somewhere to the left of
    $D_{k}$. We apply Lemma~\ref{lemma:reverse}
    (Case~\ref{enum:reverse-1} or~\ref{enum:reverse-3}) with $A =
    D_{k-1}$, $B = D_{k}$, and $Z = S_{k}$.
  \end{enumerate}

  We next show that the left neighbor of~$D_{k-1}$ is~$S_{k}$. Assume
  this is not the case.  If $S_{k}$ appears somewhere to the left
  of~$D_{k-1}$, apply Lemma~\ref{lemma:reverse}
  (Case~\ref{enum:reverse-1} or~\ref{enum:reverse-3}) with $A =
  D_{k-1}$, $B$ the left neighbor of~$D_{k-1}$, and $Z = S_{k}$.  If,
  on the other hand, $S_{k}$ appears to the right of~$D_{k}$, apply
  Lemma~\ref{lemma:reverse} (Case~\ref{enum:reverse-2}
  or~\ref{enum:reverse-4}) with $A = S_{k}$, $B$ the left neighbor
  of~$S_{k}$, and $Z = D_{k-1}$. (Note that in this case $B$ might
  be~$D_{k}$.)
\end{proof}

\begin{theorem}
  \label{theorem:optimal}
  Let $\mathcal{D}$ be a linear case instance of $n$
  disks~$D_1,\ldots,D_n$ of sizes~$d_{1} \geq d_2 \geq \dots \geq
  d_{n}$. If $n$ is even, then the following ordering is optimal:
  \[  \dots, D_{n-5}, D_5, D_{n-3}, D_3, D_{n-1}, D_{1}, D_{n}, D_{2},
  D_{n-2}, D_{4}, D_{n-4},D_{6}, \dots
  \]
  For odd~$n$, the median disk needs to be appended at the end of the
  sequence with the larger size difference.
\end{theorem}
\begin{proof}
  Let $\mathcal{D}$ be in the given ordering, and assume a better
  ordering $\mathcal{D'}$ exists. We can modify the disk sizes slightly
  so as to make them unique while keeping~$\mathcal{D'}$ better
  than~$\mathcal{D}$. But then we have a contradiction to
  Lemma~\ref{lemma:ordering}. If~$n$ is odd, then the only
  possible placements of the median disk are the left end and
  the right end, so choosing the end with the larger size difference
  gives the optimal solution.
\end{proof}

\section{\NP-Hardness of the General Case}
\label{section:generalCase}

Let us denote the decision version of our problem
as~\coinsonaline. Its input is a set of disks with rational radii and
a rational number~$\delta > 0$, the question is whether there is a feasible
placement of the disks with span at most~$\delta$.

\begin{theorem}
  \label{theorem:nphard}
  \coinsonaline is \NP-hard, even when the ratio of the largest and
  smallest disk size is bounded by six and when all numbers are given
  in unary notation.
\end{theorem}

Our proof is by reduction from
\threepartition~\cite[Problem~SP15]{GareyJohnson:1979}.  An instance
of \threepartition consists of $3m$~integers $\mathcal{A} =
a_1,\dots,a_{3m}$ and another integer~$B$, with $\sum_{i=1}^{3m} a_i =
mB$ and $\nicefrac{B}{4}<a_i<\nicefrac{B}{2}$ for all~$i$.
\threepartition decides if there is a partition of $\mathcal{A}$ into
$m$ three-element groups $A_1,\dots,A_m$ such that $\sum_{a\in A_i} a
= B$ for each group~$A_{i}$.

Given a \threepartition instance $(\mathcal{A}, B)$, we construct a
family~$\mathcal{D}$ of $12m+11$ disks, as follows:
\begin{itemize}
\item $m+1$ disks of size $1$, we will refer to these disks as
  \emph{outer frame disks};
\item $4(m+1)$ disks of size $s_0=\nicefrac{33}{100} = 0.33$, we will
  refer to these disks as \emph{inner frame disks};
\item $2(m+1)$ disks of size
  $s_1=\nicefrac{s_0}{1+s_0}=\nicefrac{33}{133}$ $(\approx 0.24812)$,
  we will refer to these disks as \emph{large filler disks};
\item $2(m+1)$ disks of size
  $s_2=\nicefrac{s_1}{1+s_1}=\nicefrac{33}{166}$ $(\approx 0.198795)$,
  we will refer to these disks as \emph{small filler disks};
\item $2$ disks of size~$s_3 = \nicefrac{1-s_0^2 - 2s_0}{4s_0} =
  \nicefrac{2311}{13200}$ $(\approx 0.175076)$, referred to as
  \emph{end disks}; 
\item $3m$ disks $D_1,\dots,D_{3m}$, referred to as \emph{partition
  disks}, where $d_i = \nicefrac{17}{99}\big(\nicefrac{3a_i}{100B} +
  \nicefrac{99}{100}\big)$.
\end{itemize}
In the following, we will identify disks by their size or type.  We
observe that all disk sizes are rational, where numerator and
denominator can be computed in time polynomial in the input size. The
radius of a disk is obtained by squaring its size.  Note that, if we multiply all
radii by the product of the denominators, then we obtain in polynomial
time an instance of our problem with integer radii.

\begin{lemma}
  \label{lemma:minimal}
  Each end disk and partition disk has size at least $s_{4} =
  \nicefrac{2261}{13200}> 0.17128$.
\end{lemma}
\begin{proof}
  Since $s_3 > s_4$, the statement is trivial for end disks.  Let
  $a_{i} \in \mathcal{A}$.  From $a_{i} > \nicefrac B4$ follows that the size
  $d_{i}$ of the corresponding partition disk is $d_{i} \geq
  \nicefrac{17}{99}\big(\nicefrac{3}{400} + \nicefrac{99}{100}\big) =
  \nicefrac{17}{99}\cdot\nicefrac{399}{400} = \nicefrac{2261}{13200}$.
\end{proof}

\subparagraph{Equivalence of the problem instances.}

We show that $\mathcal{D}$ has a placement with span~$2(m+1)$ if and
only if~$(\mathcal{A}, B)$ is a Yes-instance of~\threepartition,
implying the \NP-hardness of~\coinsonaline.

The $m+1$ outer frame disks alone already require a span of~$2(m+1)$,
so no better span is possible.  A placement of all disks
of~$\mathcal{D}$ with span~$2(m+1)$ therefore implies that consecutive
outer frame disks touch, and that all remaining disks fit into the
space under these outer frame disks.

Let's call the $m$~spaces between two consecutive (and touching) outer
frame disks \emph{gaps}.  The space to the left of the leftmost outer
frame disk is called the \emph{left end}, the \emph{right end} is
defined symmetrically.

\begin{lemma}
  \label{lemma:pattern}
  There is only one pattern of frame and filler disks (ignoring end
  disks and partition disks) that has span~$2(m+1)$.
\end{lemma}
The pattern is shown in Figure~\ref{construction:bigPicture}.  Each
gap contains eight disks of sizes
$s_2,\,s_1,\,s_0,\,s_0,\,s_0,\,s_0,\,s_1,\,s_2$; see
Figure~\ref{construction:detail}. The left end contains four disks of
sizes $s_0, s_0, s_1, s_2$, the right end contains disks of
sizes~$s_2, s_1, s_0, s_0$.
\begin{figure}[t]
  \begin{subfigure}[b]{.5\linewidth}
    \centerline{\raisebox{3mm}{\includegraphics[width=\linewidth]{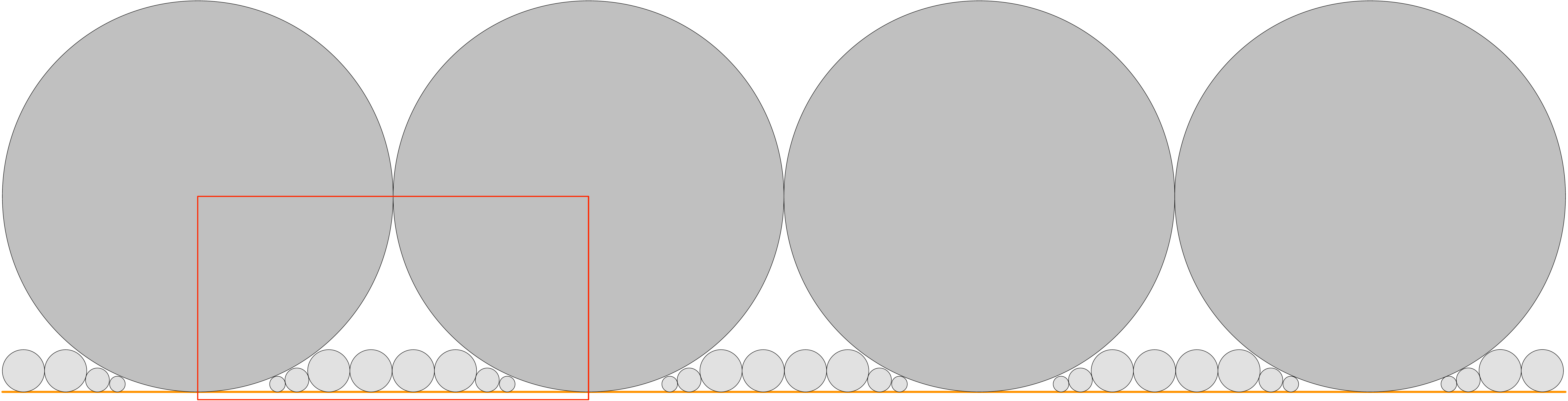}}}
    \caption{The overall picture for $m=3$.}
    \label{construction:bigPicture}
  \end{subfigure}
  \hfill
  \begin{subfigure}[b]{.45\linewidth}
    \centerline{\includegraphics[clip,trim={0 0 0 8cm},width=\linewidth]{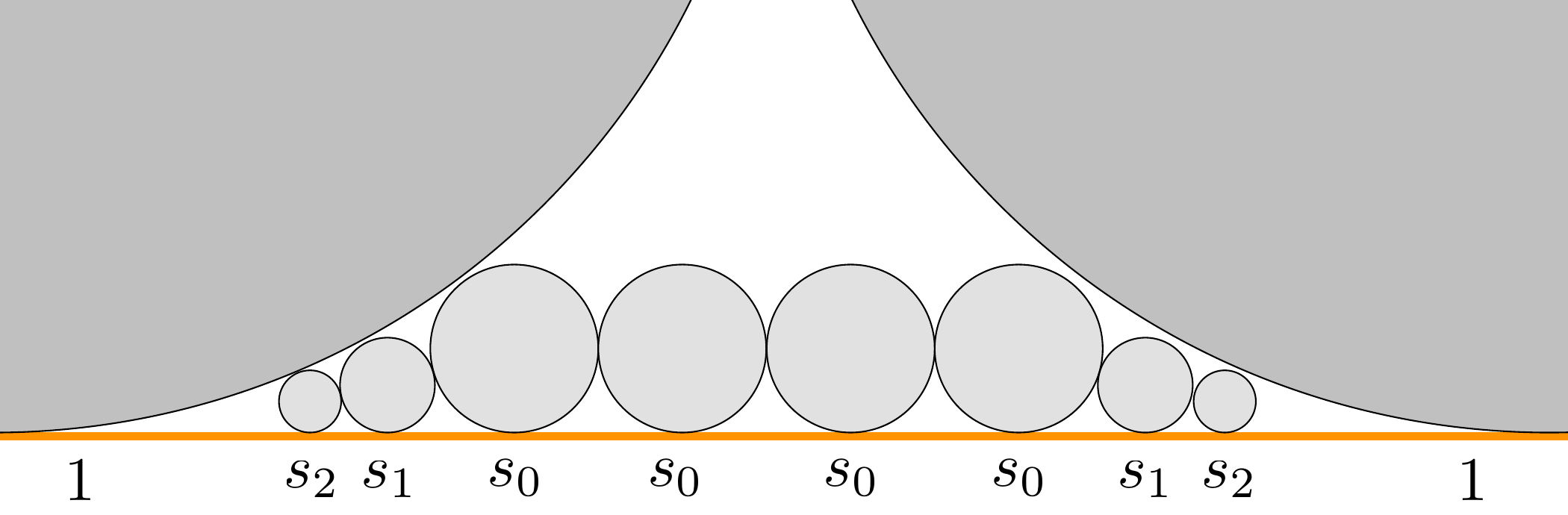}}
    \caption{The frame and filler disks inside a gap.}
    \label{construction:detail}
  \end{subfigure}
  \caption{The unique pattern of span $2(m+1)$ in
    Lemma~\ref{lemma:pattern}.}
  \label{figure:construction}
\end{figure}

Lemma~\ref{lemma:pattern} follows from the following observations
about a placement of span $2(m+1)$:
\begin{enumerate}[label=(\Alph*),leftmargin=2.4em]
\item\label{enum:5inner} A gap cannot contain five inner frame disks,
  as the total footpoint distance of the sequence $1, s_0, s_0, s_0,
  s_0, s_0, 1$ is $4s_0 + 8 s_0^2 = 2.1912$, implying that the outer
  frame disks do not touch.
\item\label{enum:3inner} The left end and the right end cannot contain
  three inner frame disks: the total footpoint distance of the
  sequence $1, s_0, s_0, s_0$ is $2s_{0} + 5s_{0}^{2} > 1.2045$, so
  this sequence does not fit in the end.
\item\label{enum:4inner} Since there are $4(m+1)$ inner frame
  disks,~\ref{enum:5inner} and~\ref{enum:3inner} imply 
  that each gap contains four inner frame disks, the left end and
  right end each contain two.
\item\label{enum:outerinner} By Lemma~\ref{lemma:gap}, a large filler
  disk fits exactly inside the space formed by a touching outer and
  inner frame disk.
\item\label{enum:2inner} Large filler disks cannot be placed between
  two inner frame disks inside a gap, as the total footpoint distance
  of the sequence $1, s_0, s_1, s_0, s_0, s_0, 1$ is $4 s_0 + 4 s_0^2
  + 4 s_1 s_2 > 2.0831$.
\item\label{enum:2large} Two large filler disks cannot be placed
  consecutively inside a gap, as the total footpoint distance of the
  sequence $1\ s_1\ s_1\ s_0\ s_0\ s_0\ s_0\ 1$ is
  $2s_1+2s_1^2+2s_0s_1+6s_0^2+2s_0>2.0965$.
\item\label{enum:1large} Only one large filler disk can appear in the
  left end and in the right end, filling the space between the outer
  and inner frame disk. Indeed, all other possibilities do not fit
  inside the end, see Table~\ref{table:DisksEnd}.
  
  \begin{table}[b]
    \caption{Impossible placements of disks in the right end.}
    \label{table:DisksEnd}
    \begin{center}
      \begin{tabular}{l||l|ll}
        \multicolumn{1}{c||}{type} & \multicolumn{1}{c|}{sequence} &
        \multicolumn{2}{c}{width} \\ 
        \hline\hline
        large filler & $1\  s_0\  s_1\  s_0$ & $2s_0+4s_0s_{1} +
        s_{0}^2$&$> 1.0964$\\ 
        & $1\  s_0\  s_0\  s_1$ & $2s_0+2s_0^{2} + 2s_0s_{1} + s_{1}^2$&$>
        1.1031$\\
        & $1\  s_1\  s_1\  s_0\  s_0$ & $2s_1+2s_1^{2} + 2s_0s_{1} + 3s_{0}^2$&$>
        1.1098$\\
        \hline
    	small filler & $1\  s_0\  s_2\  s_0$ &
  	$2s_0 + 4s_0s_{2} + s_{0}^2$ & $> 1.0313$\\
	& $1\  s_0\  s_0\  s_2$ &
	$2s_0 + 2s_0^{2} + 2s_0s_{2} + s_{2}^2$ & $> 1.0485$\\
  	& $1\  s_1\  s_2\  s_0\  s_0$ &
    	$2s_1 + 2s_1s_{2} + 2s_2s_0 + 3s_{0}^2$ & $> 1.0528$ \\
  	& $1\  s_2\  s_2\  s_1\  s_0\  s_0$ &
   	$2s_2+2s_2^{2} + 2s_2s_{1} + 2s_{1}s_{0} + 3s_{0}^2$ & $> 1.0657$\\
      \end{tabular}
    \end{center}
  \end{table}

\item\label{enum:2mlarge} Since there are $2(m+1)$ large filler disks,
  \ref{enum:outerinner}, \ref{enum:2inner}, \ref{enum:2large}, and
  \ref{enum:1large} imply that each gap contains two large filler
  disks, while the left end and right end both contain one.  Each
  large filler disk is positioned between the outer and the inner
  frame disk.
\item\label{enum:smallspace} By Lemma~\ref{lemma:gap}, a small filler
  disk fits exactly inside the space formed by an outer frame
  disk touching a large filler disk.
\item\label{enum:2small} A gap contains at most two small filler
  disks, each filling the space between the outer frame disk and the
  large filler disk. All other possibilities do not fit inside the
  gap, see Table~\ref{table:DisksGaps}.

  \begin{table}[t]
    \caption{Impossible placements of small filler disks in a gap.}
    \label{table:DisksGaps}
    \begin{center}
      \begin{tabular}{l|l@{\hspace{1mm}}l}
        \multicolumn{1}{c|}{sequence} & \multicolumn{2}{c}{total
          footpoint distance} \\ 
        \hline
        $1\  s_0\  s_2\  s_0\  s_0\  s_0\  1$ &
        $4s_0+4s_0s_2+4s_0^2$ & $>2.0180$ \\
        $1\  s_1\  s_2\  s_0\  s_0\  s_0\  s_0\  1$ &
        $2s_1+2s_1s_2+2s_2s_0+6s_0^2+2s_0$ & $>2.0395$ \\
        $1\  s_2\  s_2\  s_1\  s_0\  s_0\  s_0\  s_0\  1$ &
        $2s_2+2s_2^2+2s_2s_1+2s_1s_0+6s_0^2+2s_0$ & $>2.0524$\\
      \end{tabular}
    \end{center}
  \end{table}
\item\label{enum:smallfillerend} The left end and the right end
  contain at most one small filler disk, filling the space between the
  outer frame disk and the large filler disk. All other possibilities
  do not fit inside the end, see Table~\ref{table:DisksEnd}.

\item\label{enum:gap2small} Since there are $2(m+1)$ small filler
  disks, \ref{enum:smallspace}, \ref{enum:2small}, and
  \ref{enum:smallfillerend} imply that each gap contains two small
  filler disks, while the left end and right end each contain
  one. Each small filler disk is positioned between an outer frame
  disk and a large filler disk.
\end{enumerate}

\begin{lemma}
  \label{lemma:partition}
  Three end/partition disks $X$, $Y$, and $Z$ fit in the three gaps
  formed by the three pairs of consecutive inner frame disks in a
  common gap if and only if $x + y + z \leq \nicefrac{17}{33}$.
\end{lemma}
\begin{proof}
  By Lemma~\ref{lemma:gap}, the largest disk that fits in the space
  between two touching disks of size~$s_{0}$ has
  size~$\nicefrac{s_0}{2}$. By Lemma~\ref{lemma:minimal}, an
  end/partition disk has size at least~$s_4 > \nicefrac{s_{0}}{2}$, so
  it does not fit entirely in this space.  It follows that the total
  footpoint distance of the sequence $1, s_{0}, x, s_{0}, y, s_{0}, z,
  s_{0}, 1$ is at least $4s_0 + 4s_{0}x + 4s_{0}y + 4s_{0}z = 4s_{0}(x
  + y + z + 1)$.  $X$, $Y$, and $Z$ fit in the prescribed manner if
  and only if this total footpoint distance is at most two, proving
  the lemma.
\end{proof}

\begin{lemma}
  \label{lemma:end}
  Placing a disk $X$ in the space between the two consecutive inner
  frame disks in the left end or the right end causes the total span
  to increase if and only if $x > s_3$.
\end{lemma}
\begin{proof}
  If $x \leq \nicefrac{s_0}{2} < s_3$, the statement follows from
  Lemma~\ref{lemma:gap}, so assume $x > \nicefrac{s_0}{2}$. Then the
  total width of the sequence~$1, s_0, x, s_0$ is $2s_0 + 4s_0x +
  s_0^{2}$. The span increases if and only if this is larger than one,
  proving the lemma.
\end{proof}

\subparagraph{A 3-partition implies small span.}

Assume that $\mathcal{A}$ can be partitioned into $m$ groups~$A_i$
such that $\sum_{a\in A_i}a = B$.  Consider a group~$A_i = (a_{i1},
a_{i2}, a_{i3})$ and let $X$, $Y$, and $Z$ be the partition disks
corresponding to $a_{i1}, a_{i2}, a_{i3}$.  Then we have
\[
x + y + z = \frac{17}{99} \Big(\frac{3\cdot (a_{i1} + a_{i2} +
  a_{i3})}{100\cdot B} + 3 \cdot \frac{99}{100}\Big)
= \frac{17}{33}.
\]
By Lemma~\ref{lemma:partition} this implies that $X$, $Y$, and~$Z$ can
be placed in a common gap in the pattern of
Figure~\ref{figure:construction} without increasing the total span.
Since there are~$m$ gaps, we can place all partition disks into the
$m$~gaps.  Finally, by Lemma~\ref{lemma:end}, we can place the two end
disks inside the left end and the right end.

\subparagraph{Small span implies a 3-partition.}

We assume now that a placement of the disks~$\mathcal{D}$ with
span $2(m+1)$ exists.  By Lemma~\ref{lemma:pattern}, the frame and
filler disks must be placed in the pattern of
Figure~\ref{figure:construction}. It remains to discuss the possible
locations of the end disks and the partition disks.
We need a number of observations about a placement of span~$2(m+1)$:
\begin{enumerate}[label=(\alph*),leftmargin=2em]
\item\label{enum:leftright1} The left end and right end can contain at
  most one end disk or partition disk, and only between the two inner
  frame disks or between the outer frame disk and the small filler
  disk, see top of Table~\ref{table:Disks}.
  
  \begin{table}[t]
      \caption{Impossible placements of end/partition disks\ldots}
      \label{table:Disks}
      \centering
      \begin{tabular}{l|l@{\hspace{3mm}}l}
        \multicolumn{3}{l}{\ldots in the right end} \\
        sequence & width \\
        \hline
        $1\  s_0\  s_0\  s_4$ &
        $2s_0+2s_0^{2} + 2s_0s_4 + s_4^2$ & $> 1.0201$\\
        $1\  s_1\  s_4\  s_0\  s_0$ &
        $2s_1+2s_{1}s_4 + 2s_0s_4 + 3s_{0}^2$ & $> 1.0209$\\
        $1\  s_2\  s_4\  s_1\  s_0\  s_0$ &
        $2s_2+2s_{2}s_4 + 2s_1s_4 + 2s_{1}s_{0} + 3s_{0}^2$ & $> 1.0411$\\
        $1\  s_0\  s_4\  s_4\  s_0$ &
        $2s_0+4s_0s_4 + 2s_4^{2} + s_{0}^2$ & $>1.0536$\\
        $1\  s_4\  s_4\  s_2\  s_1\  s_0\  s_0$ &
        $2s_4+2s_4^{2} + 2s_{2}s_4 + 2s_1s_{2} + 2s_{1}s_{0} + 3s_{0}^2$ &
        $> 1.0584$\\
        $1\  s_4\  s_2\  s_1\  s_0\  s_4\  s_0$ &
        $2s_4+ 2s_{2}s_4 + 2s_1s_{2} + 2s_{1}s_{0} + 4s_{0}s_4 + s_{0}^2$ &
        $> 1.0080$ \\
        \hline 
        \multicolumn{3}{l}{ } \\
        
        \multicolumn{3}{l}{\ldots in a gap} \\
        sequence & total footpoint distance \\ 
        \hline
        $1\  s_1\  s_4\  s_0\  s_0\  s_0\  s_0\  1$ &
        $2s_1+2s_1s_4 + 2s_{0}s_4 + 6s_{0}^2 + 2s_{0}$ & $> 2.0076$\\
        $1\  s_2\  s_4\  s_1\  s_0\  s_0\  s_0\  s_0\  1$ &
        $2s_2+2s_2s_4 + 2s_4s_{1} + 2s_1s_0 + 6s_{0}^2 + 2s_{0}$ & $> 2.0278$\\
        $1\  s_0\  s_4\  s_4\  s_0\  s_0\  s_0\  1$ &
        $4s_0+4s_0s_4 + 2s_4^{2} + 4s_{0}^2$ & $> 2.0403$\\
        $1\  s_4\  s_4\  s_2\  s_1\  s_0\  s_0\  s_0\  s_0\  1$ &
        $2s_4+2s_4^{2} + 2s_4s_{2} + 2s_{2}s_{1} + 2s_1s_0 +
        6s_{0}^2 + 2s_{0}$ & $> 2.0451$\\
        $1\  s_4\  s_2\  s_1\  s_0\  s_4\  s_0\  s_4\  s_0\  s_0\  1$ &
        $2s_4+ 2s_4s_{2} + 2s_{2}s_{1} + 2s_1s_0 + 8s_{0}s_4 +
        2s_{0}^2 + 2s_{0}$ & $> 2.0030$\\
        $1\  s_4\  s_2\  s_1\  s_0\  s_4\  s_0\  s_0\  s_0\  s_1\  s_2\  s_4\  1$
        & $4s_4+ 4s_4s_{2} + 4s_{2}s_{1} + 4s_1s_0 + 4s_{0}s_4 + 4s_{0}^2$
        & $> 2.0078$\\
      \end{tabular}
      \vspace{.5em}
  \end{table}
  
\item\label{enum:gapatmost3} A gap can contain at most three partition
  disks or end disks. If a gap contains three such disks, each has to
  appear between two inner frame disks, see bottom of 
  Table~\ref{table:Disks}.
\item\label{enum:gap3} Since there are $3m+2$ end and partition disks,
  \ref{enum:leftright1} and~\ref{enum:gapatmost3} imply that each gap
  contains three such disks, while the left end and right end each
  contain one.
\item\label{enuum:sizes3} By \ref{enum:leftright1} and
  Lemma~\ref{lemma:end}, the left end and the right end can contain
  only disks of size at most~$s_3$. We can assume that these are the
  two end disks (otherwise, swap them with an end disk).
\item\label{enum:partyes} Consider a gap. It contains exactly three
  partition disks $X$, $Y$, and $Z$. By Lemma~\ref{lemma:partition},
  we have $x + y + z \leq \nicefrac{17}{33}$.  Let $a, b, c$ be the
  elements of~$\mathcal{A}$ corresponding to $X$, $Y$, and~$Z$. Then
  we have
  \[
  x + y + z = \frac{17}{99} \Big(\frac{3\cdot (a + b + c)}{100\cdot B} +
  3 \cdot \frac{99}{100}\Big) \leq \frac{17}{33},
  \]
  which implies $a + b + c \leq B$.  It follows that we have
  partitioned the elements of~$\mathcal{A}$ into~$m$ groups~$A_{1},
  A_2, \dots, A_m$ with $\sum_{a \in A_i}a \leq B$.  Since $\sum_{a
    \in \mathcal{A}}a = mB$, we must have $\sum_{a \in A_i}a = B$ for
  each~$i$, so $(\mathcal{A}, B)$ is a Yes-instance of
  \threepartition.
\end{enumerate}
This concludes the proof of Theorem~\ref{theorem:nphard}, noting that
by Lemma~\ref{lemma:minimal} all disks have size at least~$s_{4} >
\nicefrac 16$.

\section{A $\nicefrac 43$-Approximation}
\label{section:approximation}

In this section, we give a \emph{greedy algorithm} and prove that it
computes a $\nicefrac 43$-approximation to the problem.

Our algorithm starts by sorting the disks $D_1, D_2, \dots, D_n$ by
decreasing size, such that $d_1 \geq d_2 \geq \dots \geq d_n$. It then
considers the disks one by one, in this order, maintaining a placement
of the disks considered so far.  Each disk~$D$ is placed as follows:
\begin{enumerate}
\item If there is a gap between two consecutive disks~$A$ and~$B$ in
  the current placement that is large enough to contain~$D$, then we
  place~$D$ in this gap, touching the \emph{smaller} one of the two
  disks~$A$ and~$B$.
\item Otherwise, let $A$ be the leftmost disk in the current placement
  (that is, the disk with the leftmost footpoint---this is not
  necessarily the disk defining the left end of the current span), and
  let $Z$ be the rightmost disk.  Since $d \leq a$, we can place~$D$
  so that it touches~$A$ from the left (candidate placement~$D_{A}$),
  and since $d \leq z$, we can place~$D$ so that it touches~$Z$ from
  the right (candidate placement~$D_{Z}$).
\item If one of the candidate placements~$D_A$ or~$D_Z$ does not
  increase the span, we place~$D$ in this way.
\item Otherwise, we place $D$ at $D_A$ if $a > z$ and at~$D_Z$ otherwise.
\end{enumerate}
The algorithm can be implemented to run in time~$O(n \log n)$ as
follows: We maintain a priority queue that stores, for each pair of
consecutive disks, the size of the largest disk that will fit between
them.  Since we are placing disks in order of decreasing size, a newly
placed disk can only touch its two neighbors, and so it will fit into
the gap if and only if its size is at most the stored gap size.  

For the analysis of the approximation factor, we can ignore all disks
that are placed after the last disk that increased the span.  Removing
these disks from the set does not change the solution computed by the
algorithm, and can only decrease the lower bound.  We will therefore
assume in the following that the final disk~$D_{n}$ is placed using
the last rule.  We also assume that $d_n = 1$.

Next, let's call a disk~$D$ \emph{large} if $d \geq 2$, and
\emph{small} otherwise.  We have the following:
\begin{lemma}
  \label{lemma:touch}
  If two small disks are consecutive in the final placement computed
  by the algorithm, then they touch each other.
\end{lemma}
\begin{proof}
  Assume, for a contradiction, that $D$ is the \emph{first} small
  disk whose placement causes two small disks to be consecutive but
  non-touching.

  If $D$ was placed by the third or fourth rule (at the left or right
  end of the sequence), it is touching its only neighbor.
  Therefore,~$D$ must have been placed in a gap between two disks~$A$
  and~$B$.  If both~$A$ and~$B$ are small, they must be touching
  (since~$D$ is the first small disk that will not touch a neighboring
  small disk).  But, by Lemma~\ref{lemma:gap}, this means that the gap
  between~$A$ and~$B$ is too small to contain a disk of size~$d \geq
  1$.  It follows that at most one of~$A$ and~$B$ is small, say~$B$.
  But then the algorithm will place~$D$ such that it touches~$B$, a
  contradiction.
\end{proof}

To prove that our algorithm achieves approximation factor
$\nicefrac{4}{3}$, 
we will need five inequalities, which we state and prove first.

\begin{lemma}
  \label{lem:inequalities}
  The following five inequalities hold.
  \begin{alignat}{2}
    \label{eq:ineq1}
    x + y - xy &\leq 1 &\qquad& \text{for } 0 < x, y \leq 1\\
    \label{eq:ineq2}
    x + y - xy &\geq \frac{3}{4} && 
    \text{for } 0 < x, y \leq 1 \text{ and } x + y \geq 1 \\
    \label{eq:ineq3}
    x + y + xy &\geq \frac{7}{9} &&
    \text{for } \frac 13 \leq x, y \leq 1 \\
    \label{eq:ineq4}
    \frac 1x + \frac 1y + 2\frac{z-1}{xy} &\geq \frac{7}{9} &&
    \text{for } 1 \leq x, y, z \leq 3
    \text{ and } (x-z)y\leq x+z \\
    \label{eq:ineq5}
    \frac{3x + y - 1}{2x + xy + 1} &\geq \frac 34 &&
    \text{for } 1 \leq x \leq \nicefrac 32
    \text{ and } 1 \leq y \leq 4.
  \end{alignat}
\end{lemma}
\begin{proof}
  We prove the inequalities separately.
  \begin{enumerate}
  \item[(\ref{eq:ineq1})] Consider the function $f_{1}(x, y) = x + y -
    xy$.  The partial derivatives of $f_{1}$ are positive for $x, y <
    1$, so $f_{1}(x, y) \leq f_{1}(1, 1) = 1$.
  \item[(\ref{eq:ineq2})] $x + y \geq 1$ implies $f_{1}(x, y) \geq
    f_{1}(x, 1-x) = x^{2} - x + 1 = (x - \nicefrac 12)^{2} + \nicefrac
    34 \geq \nicefrac 34$.
  \item[(\ref{eq:ineq3})] Consider the function $f_{2}(x, y) = x +
    y + xy$.  Both partial derivatives of~$f_{2}$ are positive for
    positive~$x, y$, so $x, y \geq \nicefrac 13$ implies $f_{2}(x, y)
    \geq f_{2}(\nicefrac 13, \nicefrac 13) = \nicefrac 79$.
  \item[(\ref{eq:ineq4})] Consider the function $f(x, y, z) =
    \nicefrac 1x + \nicefrac 1y + \nicefrac{2(z-1)}{xy}$. The partial derivatives
    for~$x$ and~$y$ are negative for $x,y,z \geq 1$, the partial
    derivative for~$z$ is positive for~$x, y > 0$.  This implies that
    the claim holds for~$y \leq \nicefrac 94$, since then $f(x, y, z)
    \geq f(3, \nicefrac 94, 1) = \nicefrac 79$.

    The constraint $(x-z)y\leq x+z$ implies $z \geq
    (1-\nicefrac{2}{y+1})x$, and so $f(x, y, z) \geq g(x, y)$, where we
    set
    \[
    g(x,y) = f\Big(x, y, \big(1-\frac{2}{y+1}\big)x\Big)
    = \frac 1x + \frac 3y - \frac 2{xy} - \frac 4{y(y+1)}.
    \]
    Since $\frac{\partial}{\partial x} g(x, y) = \nicefrac{2-y}{x^{2}y} <
    0$ for $y > \nicefrac 94$, we have
    \[
    g(x, y) \geq g(3, y) = \frac 13 + \frac 4{y+1} - \frac 5{3y}.
    \]
    We have $\frac{\partial}{\partial y}g(3, y) =
    -\frac{7y^{2}-10y-5}{3y^{2}(y+1)^{2}} < 0$ for $y > \nicefrac 94$,
    and so $g(3, y) \geq g(3, 3) = f(3, 3, \nicefrac 32) = \nicefrac
    79$.
    
  \item[(\ref{eq:ineq5})] Consider the function~$h(x, y) = 3x + 2y
    - \nicefrac{3xy}{2}$ over the domain~$1 \leq x \leq \nicefrac 32$ and
    $1 \leq y \leq 4$.  For fixed~$y$, the function~$h(x, y)$ is
    linear in~$x$, so $h(x, y)\geq \min\big\{h(1,y), \,h(\nicefrac
    32,y)\big\}$.  We have $h(1,y) = 3 + \nicefrac{y}{2} \geq
    \nicefrac 72$ and $h(\nicefrac 32, y) = \nicefrac 92 - \nicefrac
    y4 \geq \nicefrac 72$, implying $h(x, y) \geq \nicefrac 72$.

    It follows that $\nicefrac{7}{4} \leq \nicefrac{h(x,y)}{2} =
    \nicefrac{3x}{2} + y - \nicefrac{3xy}{4}$, so $\nicefrac{3x}{2} + y \geq
    \nicefrac{3xy}{4} + \nicefrac{7}{4}$, and therefore
    \[
    3x + y - 1 = \frac 32 x + \big(\frac 32 x + y\big) - 1
    \geq \frac{3}{2}x + \frac{3}{4}xy + \frac{3}{4}
    = \frac{3}{4}\big(2x + xy + 1\big).\qedhere
    \]
  \end{enumerate}
\end{proof}

We now associate with each disk a \emph{support interval}.  The
support interval of a disk~$A$ is the interval $[\fA - 2a + 1, \fA +
  2a - 1]$. Since $0 \leq (a-1)^{2} = a^{2} -2a + 1$, we have $2a-1
\leq a^2$, and so the support interval of a disk lies within the
disk's span, see Figure~\ref{figure:support}.
\begin{figure}[t]
  \centerline{\includegraphics[width=.5\linewidth]{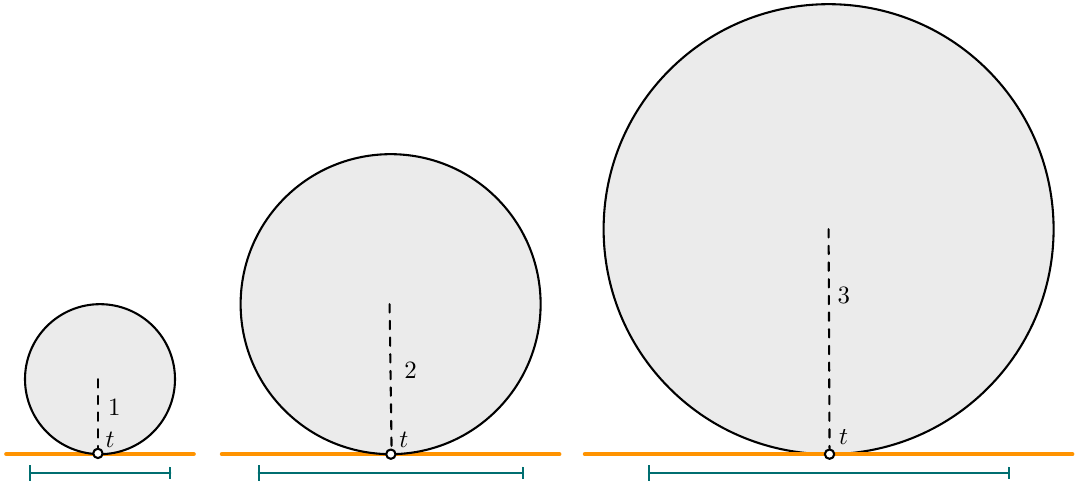}}
  \caption{Support of three disks of radius $1$, $2$ and $3$
    respectively.}
  \label{figure:support}
\end{figure}

\begin{lemma}
  \label{lemma:support}
  In any feasible placement of disks of size at least one, the open
  support intervals of the disks are disjoint.
\end{lemma}
\begin{proof}
  Consider two consecutive disks of size~$a$ and~$b$. Their footpoints
  are at distance at least~$2ab$.  The two support intervals cover $2a
  - 1 + 2b - 1$ of  this distance.  By Ineq.~(\ref{eq:ineq1}), we have
  $\nicefrac{2a + 2b  - 2}{2ab} = \nicefrac  1b + \nicefrac 1a -
  \nicefrac 1{ab} \leq 1$, and so the support intervals do not overlap.
\end{proof}

Lemma~\ref{lemma:support} implies that the total length of the support
intervals is a lower bound for the span of a family of disks. We will
show that our greedy algorithm computes a solution where 
the support intervals cover at least~$\nicefrac 34$ of the span,
implying approximation factor~$\nicefrac 43$.

Consider a pair of two consecutive disks~$A$ and~$B$ placed by the
algorithm, and let $G$ be the (imaginary) largest disk that can be
placed in the gap between~$A$ and~$B$.  Since $D_n$ was not placed in
this gap, we have~$g < 1$.  By Lemma~\ref{lemma:footpoint}, we
have~$\fA\fB = \fA\fG + \fG\fB = 2ag + 2gb = 2g(a+b)$.

Consider first the case where~$A$ and~$B$ touch.
Lemma~\ref{lemma:gap} gives $\nicefrac 1g = \nicefrac 1a + \nicefrac 1b$.
The support intervals cover $2a + 2b - 2$ of the footpoint
distance~$2ab$, so the ratio is $\nicefrac 1a + \nicefrac 1b - \nicefrac{1}{ab} \geq \nicefrac 34$
by Ineq.~(\ref{eq:ineq2}).

Now suppose that $A$ and $B$ do not touch.  By
Lemma~\ref{lemma:touch}, this means at least one of the disks is
large, say~$A$, that is~$a \geq 2$. The footpoint distance~$\fA\fB$ is
$2g(a+b)\leq 2(a+b)$, and the support intervals cover $2a + 2b - 2$ of
this distance, so the ratio~is
\[
\frac{2a + 2b - 2}{2g(a+b)} \geq
\frac{a + b - 1}{a+b} = 1 - \frac{1}{a+b}.
\]
If~$a \geq 3$ or~$b \geq 2$, we already have $1 - \nicefrac{1}{a+b} \geq
\nicefrac 34$, and this bound is good enough.

It remains to consider the situation when $2 \leq a < 3$ and $1 \leq b
\leq 2$.  Breaking symmetry, we assume without loss of generality that
$B$ is to the right of~$A$. We denote the first disk to the right
of~$A$ that is touching~$A$ as~$D$.  By the nature of our algorithm,
when~$B$ was placed, it was placed inside the space between~$A$
and~$D$ (possibly, other disks were already present in this space at
that time).  Since~$B$ does not touch~$A$, the disk~$D$ must be
smaller than~$A$, that is $1 \leq d \leq a < 3$.

We analyze the entire interval~$[\fA, \fD]$ as a whole.  Since $A$
and~$D$ touch, the length of this interval is~$2ad$.  In between~$A$
and~$D$, some $k \geq 1$ disks have been placed, with~$B$ being the
leftmost of these.

We first consider the case $k \geq 2$. The total length of the support
intervals in the interval~$\fA\fD$ is at least $2a - 1 + 2d - 1 + 2k
\geq 2(a + d + 1)$.  The distance~$\fA\fD$ is~$2ad$, and by
Ineq.~(\ref{eq:ineq3})
\[
\frac{2(a+d+1)}{2ad} = \frac 1a + \frac 1d + \frac 1{ad}
\geq \frac{7}{9} > \frac{3}{4}.
\]

In the second case, $B$~is the only disk between~$A$ and~$D$.  This
means that~$B$ touches~$D$.  The total support interval length in the
interval~$\fA\fD$ is
\[
2a-1 + 4b - 2 + 2d - 1 = 2a + 4b +2d - 4.
\]  
Let $G$ be the largest disk that fits in the gap between~$A$
and~$B$.  Its size is determined by the equality $2ag + 2gb + 2bd =
2ad$, so $g = \nicefrac{(a-b)d}{a+b}$.  Since $D_n$ was not placed in this
gap, we have $g < 1$, and so $(a-b)d < a+b$.  Then
Ineq.~(\ref{eq:ineq4}) implies
\[
\frac{2a + 4b +2d - 4}{2ad} =
\frac 1a + \frac 1d + \frac{2(b-1)}{ad} \geq \frac 79 > \frac 34.
\]

To complete the proof, we need to argue about the part of the span
that does not lie between two footpoints, in other words, the two
intervals between the left wall (defined by the leftmost point on any
disk) and the leftmost footpoint, and between the rightmost footpoint
and the right wall.  Recall that we assumed that placing~$D_n$
increased the total span.  This implies that~$D_n$ was placed using
the algorithm's last rule and therefore touches one of the two walls,
let's say the right wall.  Let~$A$ and~$B$ be the two leftmost disks
(in footpoint order), and let~$Y$ and~$Z$ be the two rightmost disks
(in footpoint order).  By assumption, $Z = D_n$ and so~$z = 1$.  Since
$D_n$ was placed using the last rule, we have $y \geq a$, and $Z$
touches~$Y$. Let us call~$G$ the (imaginary) largest disk that would
fit into the space between the left wall and~$A$.  Since $D_n$ was not
placed in this position, we have~$g < 1$.  Note that the left wall is
at coordinate $\fG - g^2$, and the right wall at coordinate~$\fZ + 1$.
We now distinguish two cases.

We first consider the case where~$a \geq \nicefrac 32$.  We then
analyze the two intervals~$[\fG - g^2, \fA]$ and $[\fY, \fZ + 1]$
together.  Their total length is $g^2 + 2ga + 2y + 1 < 2y + 2a + 2$,
and the support intervals of $A$, $Y$, and~$Z$ cover $2a-1 + 2y-1 + 2
= 2y + 2a$ of this.  The ratio is
\[
\frac{2y + 2a}{2y + 2a + 2} = 1 - \frac{1}{y + a + 1} \geq 1 -
\frac{1}{4} = \frac{3}{4} \qquad
\text{since} \quad y \geq a \geq \nicefrac 32.
\]

In the second case we have~$a < \nicefrac 32$.  Then $B$ must be
touching~$A$.  This is true if $b \geq a$, because then $A$ was placed
later than~$B$ using the third rule.  When $b < a$, then it follows
from Lemma~\ref{lemma:touch}.  The distance between~$\fG - g^{2}$ and
$\fB$ is then~$g^{2} + 2ag + 2ab \leq 2ab + 2a + 1 \leq 3b +
4$. Since~$B$ fits inside the span, we must have $b^{2} \leq 3b + 4$,
which solves to~$-1 \leq b \leq 4$.

We now analyze the intervals~$[\fG - g^2, \fB]$ and $[\fY, \fZ + 1]$
together.  Their total length is
\[
g^2 + 2ga + 2ab + 2y + 1 < 2y + 2a + 2ab + 2,
\]
while the support intervals of $A$, $B$, $Y$, and~$Z$ cover
\[
4a - 2 + 2b - 1 + 2y - 1 + 2 = 2y + 4a + 2b -2.
\]
Since $y \geq a$, we can lower-bound the ratio using Ineq.~(\ref{eq:ineq5})
\[
\frac{2y + 4a + 2b - 2}{2y + 2a + 2ab + 2}
\geq \frac{6a + 2b - 2}{4a + 2ab + 2}
= \frac{3a + b - 1}{2a + ab + 1} \geq \frac 34.
\]
Note that in this second case we have used the interval~$[\fA, \fB]$
to help bound the coverage of the two end intervals.  This could be a
problem if the same interval was also needed to help bound a larger
interval of the form~$[\fA, \fC]$, where $A$ and $C$ touch and~$B$ was
inserted into this interval later.  But note that we needed to
analyze~$[\fA,\fC]$ as a whole only if~$c < 3$.  Since $a < \nicefrac
32$, no disk of size one would then fit into the gap between~$A$
and~$C$, so this situtation cannot occur.

This completes the proof of the following theorem.
\begin{theorem}
  \label{theorem:greedy}
  The greedy algorithm computes a $\nicefrac 43$-approximation in
  time~$O(n \log n)$.
\end{theorem}

\section{Conclusions}

Our best approximation algorithm achieves an approximation factor
of~$\nicefrac 43$.  We were unable to find a polynomial time approximation
scheme, so it would be natural to try to prove that the problem is
\APX-hard.  This, however, seems unlikely to be true, for the same
reasons as outlined by D\"urr et al.~\cite{Duerr2017}: The ideas they
present appear to transfer to our problem, and would lead to an
$2^{O(\log^{O(1)}n)}$ algorithm with approximation
factor~$(1+\varepsilon)$.  \APX-hardness, on the other hand, would
imply that for some~$\varepsilon>0$ this approximation problem
is~\NP-hard, implying subexponential algorithms for~\NP.

\section*{Acknowledgments}

We thank Peyman Afshani and Ingo van Duin for helpful discussions
during O.C.'s visit to Madalgo in 2016.  We also thank all
participants at the Korean Workshop on Computational Geometry in
W\"urzburg~2016.

\bibliographystyle{abbrvurl}
\bibliography{coins}
\end{document}